\documentclass[oneside,english]{amsart}
\usepackage[T1]{fontenc}
\usepackage[latin9]{inputenc}
\usepackage{color}
\usepackage{babel}

\usepackage{float}
\usepackage{amsthm}
\usepackage{graphicx}
\usepackage{amssymb}
\usepackage[unicode=true, pdfusetitle,
 bookmarks=true,bookmarksnumbered=true,bookmarksopen=true,bookmarksopenlevel=2,
 breaklinks=false,pdfborder={0 0 1},backref=false,colorlinks=true]
 {hyperref}

\makeatletter
\numberwithin{equation}{section}
\numberwithin{figure}{section}
\theoremstyle{plain}
\newtheorem{thm}{Theorem}

\@ifundefined{showcaptionsetup}{}{%
 \PassOptionsToPackage{caption=false}{subfig}}
\usepackage{subfig}
\makeatother

\begin{document}

\title[A numerical implementation of the perturbation theory ...]{A numerical and symbolical approximation of the Nonlinear Anderson
Model}

\author{Yevgeny Krivolapov}

\address{Physics Department, Technion - Israel Institute of Technology, Haifa
32000, Israel.}

\email{evgkr@tx.technion.ac.il}

\author{Shmuel Fishman}

\address{Physics Department, Technion - Israel Institute of Technology, Haifa
32000, Israel.}

\email{fishman@physics.technion.ac.il}

\author{Avy Soffer}

\address{Mathematics Department, Rutgers University, New-Brunswick, NJ 08903,
USA.}

\email{soffer@math.rutgers.edu}

\begin{abstract}
A modified perturbation theory in the strength of the nonlinear term
is used to solve the Nonlinear Schrödinger Equation with a random
potential. It is demonstrated that in some cases it is more efficient
than other methods. Moreover we obtain error estimates. This approach
can be useful for the solution of other nonlinear differential equations
of physical relevance.
\end{abstract}

\date{December 19, 2009}

\maketitle

\section{Introduction}

We consider the problem of dynamical localization of waves in a Nonlinear
Schrödinger Equation (NLSE) \cite{Sulem1999} with a random potential
term on a lattice:\begin{equation}
i\partial_{t}\psi=-J\left[\psi\left(x+1\right)+\psi\left(x-1\right)\right]+\varepsilon_{x}\psi+\beta\left\vert \psi\right\vert ^{2}\psi,\label{eq:NLSE}\end{equation}
where $\psi=\psi\left(x,t\right),$ $x\in\mathbb{Z};$ and $\left\{ \varepsilon_{x}\right\} $
is a collection of i.i.d. random variables uniformly distributed in
the interval $\left[-\frac{W}{2},\frac{W}{2}\right]$.

The NLSE was derived for a variety of physical systems under some
approximations. It was derived in classical optics where $\psi$ is
the electric field by expanding the index of refraction in powers
of the electric field keeping only the leading nonlinear term \cite{Agrawal2007}.
For Bose-Einstein Condensates (BEC), the NLSE is a mean field approximation
where the term proportional to the density $\beta|\psi|^{2}$ approximates
the interaction between the atoms. In this field the NLSE is known
as the Gross-Pitaevskii Equation (GPE) \cite{Dalfovo1999,Pitaevskii2003,Leggett2001,Pitaevskii1961,Gross1961,Pitaevskii1963}.
Recently, it was rigorously established, for a large variety of interactions
and of physical conditions, that the NLSE (or the GPE) is exact in
the thermodynamic limit \cite{Erdos2007,Lieb2002}. Generalized mean
field theories, in which several mean-fields are used, were recently
developed \cite{Cederbaum2003,Alon2005}. In the absence of randomness
\eqref{eq:NLSE} is completely integrable. For repulsive nonlinearity
$\left(\beta>0\right)$ an initially localized wavepacket spreads,
while for attractive nonlinearity $\left(\beta<0\right)$ solitons
are found typically \cite{Sulem1999}.

For $\beta=0$ this equation reduces to the Anderson model \cite{Anderson1958},\begin{equation}
i\partial_{t}\psi=-J\left[\psi\left(x+1\right)+\psi\left(x-1\right)\right]+\varepsilon_{x}\psi.\end{equation}
It is well known that in 1D in the presence of a random potential
with probability one all the states are exponentially localized \cite{Anderson1958,Ishii1973,Lee1985,Lifshits1988}.
Consequently, diffusion is suppressed and in particular a wavepacket
that is initially localized will not spread to infinity. This has
been very recently extended to the many-body particle system \cite{Aizenman2008,Basko2006,Basko2007}.
This is the phenomenon of Anderson localization. In 2D it is known
heuristically from the scaling theory of localization \cite{Abrahams1979,Lee1985}
that all the states are localized, while in higher dimensions there
is a mobility edge that separates localized and extended states. This
problem is relevant for experiments in nonlinear optics, for example
disordered photonic lattices \cite{Schwartz2007}, where Anderson
localization was found in presence of nonlinear effects as well as
experiments on BECs in disordered optical lattices \cite{Gimperlein2005,Lye2005,Clement2005,Clement2006,Sanchez-Palencia2007,Billy2008,Fort2005,Akkermans2008,Paul2007}.
The interplay between disorder and nonlinear effects leads to new
interesting physics \cite{Fort2005,Akkermans2008,Bishop1995,Rasmussen1999,Kopidakis1999,Kopidakis2000}.
In spite of the extensive research, many fundamental problems are
still open, and in particular, it is not clear whether in one dimension
(1D) Anderson localization can survive the effects of nonlinearities
(see however \cite{Wang2008,Wang2008a}).

A natural question is whether a wave packet that is initially localized
in space will indefinitely spread for dynamics controlled by (\ref{eq:NLSE}).
A simple argument indicates that spreading will be suppressed by randomness.
If unlimited spreading takes place the amplitude of the wave function
will decay since the $L^{2}$ norm is conserved. Consequently, the
nonlinear term will become negligible and Anderson localization will
take place as a result of the randomness. Contrary to this intuition,
based on the smallness of the nonlinear term resulting from the spread
of the wave function, it is claimed that for the kicked-rotor a nonlinear
term leads to delocalization if it is strong enough \cite{Shepelyansky1993}.
It is also argued that the same mechanism results in delocalization
for the model \eqref{eq:NLSE} with sufficiently large $\beta$, while,
for weak nonlinearity, localization takes place \cite{Shepelyansky1993,Pikovsky2008}.
Therefore, it is predicted in that work that there is a critical value
of $\beta$ that separates the occurrence of localized and extended
states. However, if one applies the arguments of \cite{Shepelyansky1993,Pikovsky2008}
to a variant of \eqref{eq:NLSE}, results that contradict numerical
solutions are found \cite{Mulansky2009,Veksler2009}. Recently, it
was rigorously shown that the initial wavepacket cannot spread so
that its amplitude vanishes at infinite time, at least for large enough
$\beta$ \cite{Kopidakis2008}. It does not contradict spreading of
a fraction of the wavefunction. Indeed, subdiffusion was found in
numerical experiments \cite{Shepelyansky1993,Kopidakis2008,Molina1998}.
In different works \cite{Molina1998,Flach2009,Skokos2009} sub-diffusion
was reported for all values of $\beta$. It was also argued that nonlinearity
may enhance discrete breathers \cite{Kopidakis1999,Kopidakis2000}.
In conclusion, it is \emph{not} clear what is the long time behavior
of a wave packet that is initially localized, if both nonlinearity
and disorder are present. The major difficulty in numerical resolution
of this question is integration of \eqref{eq:NLSE} to large time.
Most researchers who run numerical simulation use a split-step method
for integration, however it is impossible to achieve convergence for
large times, and therefore some heuristic arguments assuming that
the numerical errors do not affect the results qualitatively, are
utilized \cite{Shepelyansky1993,Flach2009}. However it is unclear
whether those arguments apply to \eqref{eq:NLSE}. The motivation
of the current work is to propose a numerical scheme based on a modified
perturbation theory developed in \cite{Fishman2008a,Fishman2009a}
which will allow integration of \eqref{eq:NLSE} and similar equations
up to large times and with some control of the error based on the
form of the remainder term obtained in \cite{Fishman2009a}.

The advantage of the perturbative method is that it provides an estimate
of the error while there is no such estimate in the split-step method.
Moreover the error in the split-step method is expected to proliferate
as a result of the nonlinearity.

In Section \ref{sec:Pert_theory} we briefly review the perturbation
theory developed in \cite{Fishman2009a}. In Section \ref{sec:The-numerical-method}
we explain the numerical scheme used to compute the different orders
in the perturbation theory. In Section \ref{sec:The-remainder} we
show how the error of the perturbation theory could be controlled.
In Section \ref{sec:Results} we present the comparison between the
perturbation theory and an exact integration. The results are summarized
in Section \ref{sec:Summary} and open problems are listed there.

In summary, this work demonstrates a numerical implementation of the
perturbation theory for \eqref{eq:NLSE} in powers of $\beta$ which
was described in \cite{Fishman2009a}, and evaluates its possible
use.

\section{\label{sec:Pert_theory}The perturbation theory}

Our goal is to analyze the nonlinear Schrödinger equation \eqref{eq:NLSE}
that could be written in the form\begin{equation}
i\partial_{t}\psi=H_{0}\psi+\beta\left\vert \psi\right\vert ^{2}\psi\label{eq:GPE}\end{equation}
 where $H_{0}$ is the Anderson Hamiltonian,\begin{equation}
H_{0}\psi\left(x\right)=-J\left[\psi\left(x+1\right)+\psi\left(x-1\right)\right]+\varepsilon_{x}\psi\left(x\right).\end{equation}
The wavefunction can be expanded using the eigenstates, $u_{m}\left(x\right)$,
and eigenvalues, $E_{m}$, of $H_{0}$ as \begin{equation}
\psi\left(x,t\right)=\sum_{m}c_{m}\left(t\right)e^{-iE_{m}t}u_{m}\left(x\right).\label{eq:expansion}\end{equation}
For the nonlinear equation the dependence of the expansion coefficients,
$c_{n}\left(t\right),$ is found by inserting this expansion into
\eqref{eq:GPE}, resulting in\begin{equation}
i\partial_{t}c_{n}=\beta\sum_{m_{1},m_{2},m_{3}}V_{n}^{m_{1}m_{2}m_{3}}c_{m_{1}}^{\ast}c_{m_{2}}c_{m_{3}}e^{i\left(E_{n}+E_{m_{1}}-E_{m_{2}}-E_{m_{3}}\right)t}\label{eq:c_n_exact}\end{equation}
 where $V_{n}^{m_{1}m_{2}m_{3}}$ is an overlap sum\begin{equation}
V_{n}^{m_{1}m_{2}m_{3}}=\sum_{x}u_{n}\left(x\right)u_{m_{1}}\left(x\right)u_{m_{2}}\left(x\right)u_{m_{3}}\left(x\right).\label{eq:overlap_int}\end{equation}
Our objective is to develop a perturbation expansion of the $c_{m}\left(t\right)$
in powers of $\beta$ and to calculate them order by order in $\beta.$
The required expansion is\begin{equation}
c_{n}\left(t\right)=c_{n}^{\left(0\right)}+\beta c_{n}^{\left(1\right)}+\beta^{2}c_{n}^{\left(2\right)}+\cdots+\beta^{N}c_{n}^{\left(N\right)}+Q_{n},\label{eq:cn_expand}\end{equation}
 where the expansion is till order $N$ and $Q_{n}$ is the remainder
term (note here $Q_{n}$ differs from one defined in \cite{Fishman2009a}).
We will assume the initial condition\begin{equation}
c_{n}\left(t=0\right)=\delta_{n0}.\end{equation}
 The equations for the two leading orders are presented in what follows.
The leading order is\begin{equation}
c_{n}^{\left(0\right)}=\delta_{n0}.\label{eq:cn0}\end{equation}
And the first order is\begin{equation}
c_{n}^{\left(1\right)}=V_{n}^{000}\left(\frac{1-e^{i\left(E_{n}-E_{0}\right)t}}{E_{n}-E_{0}}\right).\end{equation}
We notice that divergence of this expansion for any value of $\beta$
may result from three major problems: the secular terms problem, the
entropy problem (i.e., factorial proliferation of terms), and the
small denominators problem. In this paper we will not discuss the
entropy problem and the problem of small denominators, since it was
done in detail in \cite{Fishman2009a}. We first show how to derive
the equations for $c_{n}\left(t\right)$ where the secular terms are
eliminated. To achieve this we replace the ansatz \eqref{eq:expansion}
by

\begin{equation}
\psi\left(x,t\right)={\displaystyle \sum\limits _{n}}c_{n}\left(t\right)e^{-iE_{n}^{\prime}t}u_{n}\left(x\right)\label{eq:expansion_prime}\end{equation}
 where\begin{equation}
E_{n}^{\prime}\equiv E_{n}^{\left(0\right)}+\beta E_{n}^{\left(1\right)}+\beta^{2}E_{n}^{\left(2\right)}+\cdots\label{eq:En_expand}\end{equation}
 and $E_{n}^{\left(0\right)}$ are the eigenvalues of $H_{0}$. We
will dub $E_{n}'$ the renormalized energies. The new equation for
the $c_{n}$ is given by

\begin{equation}
i\partial_{t}c_{n}=\left(E_{n}^{\left(0\right)}-E_{n}^{\prime}\right)c_{n}+\beta\sum_{m_{1}m_{2}m_{3}}V_{n}^{m_{1}m_{2}m_{3}}c_{m_{1}}^{\ast}c_{m_{2}}c_{m_{3}}e^{i\left(E_{n}^{\prime}+E_{m_{1}}^{\prime}-E_{m_{2}}^{\prime}-E_{m_{3}}^{\prime}\right)t}.\label{eq:diff_eq}\end{equation}
Inserting expansions \eqref{eq:cn_expand} and \eqref{eq:En_expand}
into \eqref{eq:diff_eq} and comparing the powers of $\beta$ \emph{without
expanding the exponent} in $\beta$, produces the following equation
for the $k-th$ order\begin{align}
i\partial_{t}c_{n}^{\left(k\right)} & =-\sum_{s=0}^{k-1}E_{n}^{\left(k-s\right)}c_{n}^{\left(s\right)}+\label{eq:a2}\\
 & +\sum_{m_{1}m_{2}m_{3}}V_{n}^{m_{1}m_{2}m_{3}}\left[\sum_{r=0}^{k-1}\sum_{s=0}^{k-1-r}\sum_{l=0}^{k-1-r-s}c_{m_{1}}^{\left(r\right)\ast}c_{m_{2}}^{\left(s\right)}c_{m_{3}}^{\left(l\right)}\right]e^{i\left(E_{n}^{\prime}+E_{m_{1}}^{\prime}-E_{m_{2}}^{\prime}-E_{m_{3}}^{\prime}\right)t}.\nonumber \end{align}
Note that the exponent is of order $O\left(1\right)$ in $\beta$,
and therefore we may choose not to expand it in powers of $\beta$.
However, it results in an expansion where both $E_{m}^{\left(l\right)}$
and $c_{n}^{\left(k\right)}$ depend on $\beta$. For the expansion
\eqref{eq:cn_expand} to be valid, both $E_{m}^{\left(l\right)}$
and $c_{n}^{\left(k\right)}$ should be $O\left(1\right)$ in $\beta$,
this is satisfied, since the RHS of \eqref{eq:a2} contains only $c_{n}^{\left(r\right)}$
such that $r<k$. Namely, this equation gives each order in terms
of the lower ones, with the initial condition of $\left.c_{n}^{\left(0\right)}\left(t\right)=\delta_{n0}\right..$
Solution of $k$ equations \eqref{eq:a2} gives the solution of the
differential equation \eqref{eq:diff_eq} to order $k$, while the
higher order terms which are obtained from this equation are meaningless
(see Appendix for the reasoning). Since, the exponent in \eqref{eq:a2}
is of order $O\left(1\right)$ in $\beta$ we can select its argument
to be of any order in $\beta$. However, for the removal of the secular
terms, as will be explained bellow, it is instructive to set the order
of the argument to be $k-1$, as the higher orders were not calculated
at this stage. Secular terms are created when there are time independent
terms in the RHS of the equation above. We eliminate those terms by
using the first two terms in the first summation on the RHS. We make
use of the fact that $c_{n}^{\left(0\right)}=\delta_{n0}$ and $c_{n}^{\left(1\right)}$
can be easily determined (see (\ref{eq:c01},\ref{eq:cn1-1})), and
used to calculate $E_{n=0}^{\left(k\right)}$ and $E_{n\neq0}^{\left(k-1\right)}$
that eliminate the secular terms in the equation for $c_{n}^{\left(k\right)},$
that is \begin{equation}
E_{n}^{\left(k\right)}c_{n}^{\left(0\right)}+E_{n}^{\left(k-1\right)}c_{n}^{\left(1\right)}=E_{n}^{\left(k\right)}\delta_{n0}+E_{n}^{\left(k-1\right)}\left(1-\delta_{n0}\right)\frac{V_{n}^{000}}{E_{n}^{\prime}-E_{0}^{\prime}},\label{eq:sec_elim}\end{equation}
where only the time-independent part of $c_{n}^{\left(1\right)}$
was used. In other words, we choose $E_{n}^{\left(k\right)}$ and
$E_{n\neq0}^{\left(k-1\right)}$ so that the time-independent terms
on the RHS of (\ref{eq:a2}) are eliminated. $E_{0}^{\left(k\right)}$
will eliminate all secular terms with $n=0,$ and $E_{n}^{\left(k-1\right)}$
will eliminate all secular terms with $n\neq0.$ In the following,
we will demonstrate this procedure for the first order.

In the first order of the expansion in $\beta$ we obtain\begin{align}
i\partial_{t}c_{n}^{\left(1\right)} & =-E_{n}^{\left(1\right)}c_{n}^{\left(0\right)}+\sum_{m_{1}m_{2}m_{3}}V_{n}^{m_{1}m_{2}m_{3}}c_{m_{1}}^{\ast\left(0\right)}c_{m_{2}}^{\left(0\right)}c_{m_{3}}^{\left(0\right)}e^{i\left(E_{n}'+E_{m_{1}}'-E_{m_{2}}'-E_{m_{3}}'\right)t}\label{eq:cn1-1}\\
 & =-E_{n}^{\left(1\right)}\delta_{n0}+V_{n}^{000}e^{i\left(E_{n}'-E_{0}'\right)t}.\nonumber \end{align}
 For $n=0$ the equation produces a secular term\begin{align}
i\partial_{t}c_{0}^{\left(1\right)} & =-E_{0}^{\left(1\right)}+V_{0}^{000}\label{eq:c01}\\
c_{0}^{\left(1\right)} & =it\cdot\left(E_{0}^{\left(1\right)}-V_{0}^{000}\right).\nonumber \end{align}
 Setting \begin{equation}
E_{0}^{\left(1\right)}=V_{0}^{000}\end{equation}
 eliminates this secular term and gives \begin{equation}
c_{0}^{\left(1\right)}=0.\end{equation}
 For $n\neq0$ there are no secular terms in this order, therefore
finally\begin{equation}
c_{n}^{\left(1\right)}=\left(1-\delta_{n0}\right)V_{n}^{000}\left(\frac{1-e^{i\left(E_{n}'-E_{0}'\right)t}}{E_{n}'-E_{0}'}\right),\label{eq:cn1}\end{equation}
where to this order $E_{n}'=E_{n}$ and $E'_{0}=E_{0}$.

The higher order terms in the perturbation theory are given by recursive
relations and due to the large number of terms which are involved
will be calculated numerically.

\section{\label{sec:The-numerical-method}The numerical method}

In order to compute the various orders in the perturbation theory
we use equation \eqref{eq:a2}, which is a recursive equation of the
orders. To compute order $k$ we have to compute all $c_{n}^{\left(l\right)}$
and $E_{n}^{\left(l\right)}$ for $l\leq k-1$. The numerical calculation
is done in two stages: at the first stage a symbolic calculation of
the expressions of all the $c_{n}^{\left(l\right)}$ and $E_{n}^{\left(l\right)}$
is performed, this has a complexity of $O\left(e^{2k}\right)$, which
is due to the increasing number of terms in each expression for $c_{n}^{\left(l\right)}$
(see \cite{Fishman2009a}). This stage does not depend neither on
the realization nor the nonlinearity strength, $\beta$. In the second
stage realizations and $\beta$ are chosen and $c_{n}^{\left(l\right)}$
and $E_{n}^{\left(l\right)}$ are calculated. This stage has a complexity
of $O\left(e^{2k}\cdot L^{k}\right)$, where $L$ is the dimension
of the lattice. The computation of this stage could be fully parallelized.

When calculating $E_{n}^{\left(l\right)}$ we encounter self-consistent
equations of the type\begin{equation}
E_{n}^{'}=f_{n}\left(\left\{ E_{m}^{'}\right\} \right),\end{equation}
where $f$ is some function, for example for the second order\begin{equation}
E_{n}^{\prime}=E_{n}^{\left(0\right)}+\beta V_{n}^{n00}\left(2-\delta_{n0}\right)-3\beta^{2}\delta_{n0}\sum_{m\neq0}\frac{\left(V_{m}^{000}\right)^{2}}{E_{m}^{\prime}-E_{0}^{\prime}}.\label{eq:en_oder2}\end{equation}
Higher order equations are required in general. We solve those equations
numerically by reinserting the LHS into the RHS, until a desired convergence
is achieved. The first iteration is done by setting $\beta=0$ at
the RHS. Basically, at each iteration an order of $\beta$ is gained
in the accuracy of the solution and since we need to know $E_{n}^{'}$
only to a desired order $N$ (see Appendix), only a small number of
iterations is needed. The $c_{n}^{\left(l\right)}$ are represented
as vectors with elements $\left(c_{n,\omega_{1}}^{\left(l\right)},c_{n,\omega_{2}}^{\left(l\right)},\ldots\right)$
identified by frequencies such that terms with same frequencies are
grouped together (by summing their amplitudes), namely, $c_{n,\omega_{k}}^{\left(l\right)}=\sum_{j}c_{n,\omega_{k},j}^{\left(l\right)}e^{-i\omega_{k}t}$,
where $\omega_{k}$ is a shared frequency. Due to the fact that most
of the amplitudes are negligible, after grouping a thresholding step
is done and terms which are smaller than $10^{-6}$ are eliminated.
The error introduced by the tresholding can be easily controlled,
since we know how many frequencies were left out. By having the vector
of frequencies and their corresponding amplitudes we can calculate
the perturbative solution at any time. Even after grouping and thresholding
the number of frequencies is growing rapidly with the order of the
expansion.

\section{\label{sec:The-remainder}The remainder of the expansion}

In order to control the solution we have to control, $Q_{n}$, the
remainder of the expansion \eqref{eq:cn_expand} that can be written
in the form\begin{equation}
c_{n}\left(t\right)=\tilde{c}_{n}+Q_{n},\label{eq:remainder_def}\end{equation}
with \begin{eqnarray}
\tilde{c}_{n} & = & \sum_{l=0}^{N}\beta^{l}c_{n}^{\left(l\right)}.\label{eq:c_tilde_def}\end{eqnarray}
It is useful to define\begin{eqnarray}
\tilde{\psi}\left(x,t\right) & = & \sum_{m}\tilde{c}_{m}u_{m}\left(x\right)e^{-iE_{m}^{'}t}\label{eq:psi_tilde_def}\end{eqnarray}
and\begin{equation}
\tilde{Q}_{n}=Q_{n}e^{-iE_{n}^{'}t}.\label{eq:Q_tilde_def}\end{equation}
Substituting \eqref{eq:remainder_def} in \eqref{eq:diff_eq} leads
to the following equation for the remainder which is expressed in
terms of \eqref{eq:c_tilde_def}, \eqref{eq:psi_tilde_def} and \eqref{eq:Q_tilde_def},
\begin{equation}
i\partial_{t}\tilde{Q}_{n}=W_{n}\left(t\right)+\sum_{m}M_{nm}\left(t\right)\tilde{Q}_{m}+\sum_{m}\bar{M}_{nm}\left(t\right)\tilde{Q}_{m}^{*}+F\left(\tilde{Q}\right)\label{eq:bootstrap}\end{equation}
where\begin{eqnarray}
W_{n}\left(t\right) & = & \left(E_{n}^{\left(0\right)}-E_{n}^{'}\right)\tilde{c}_{n}e^{-iE_{n}^{'}t}-i\left(\partial_{t}\tilde{c}_{n}\right)e^{-iE_{n}^{'}t}\\
 & + & \beta\sum_{x}u_{n}\left(x\right)\left|\tilde{\psi}\left(x\right)\right|^{2}\tilde{\psi}\left(x\right)\nonumber \end{eqnarray}
is the inhomogeneous term,\begin{eqnarray}
M_{nm}\left(t\right) & = & E_{n}^{\left(0\right)}\delta_{nm}+2\beta\sum_{x}u_{n}\left(x\right)\left|\tilde{\psi}\left(x\right)\right|^{2}u_{m}\left(x\right)\end{eqnarray}
and\begin{equation}
\bar{M}_{nm}\left(t\right)=\beta\sum_{x}u_{n}\left(x\right)\left(\tilde{\psi}\left(x\right)\right)^{2}u_{m}\left(x\right).\end{equation}
determine the linear terms, while the nonlinear term is,\begin{eqnarray}
F\left(\bar{Q}\right) & = & \beta\sum_{x}u_{n}\left(x\right)\tilde{\psi}^{*}\left(x\right)\left(\sum_{m}\tilde{Q}_{m}u_{m}\left(x\right)\right)^{2}\\
 & + & 2\beta\sum_{x}u_{n}\left(x\right)\tilde{\psi}\left(x\right)\left|\sum_{m}\tilde{Q}_{m}u_{m}\left(x\right)\right|^{2}\nonumber \\
 & + & \beta\sum_{x}u_{n}\left(x\right)\left|\sum_{m}\tilde{Q}_{m}u_{m}\left(x\right)\right|^{2}\left(\sum_{m}\tilde{Q}_{m}u_{m}\left(x\right)\right).\nonumber \end{eqnarray}
The linear part of \eqref{eq:bootstrap} is given by\begin{equation}
i\partial_{t}\tilde{Q}_{n}^{lin}=W_{n}\left(t\right)+\sum_{m}M_{nm}\left(t\right)\tilde{Q}_{m}^{lin}+\sum_{m}\bar{M}_{nm}\left(t\right)\left(\tilde{Q}_{m}^{lin}\right)^{*}.\label{eq:bootstrap_lin}\end{equation}

Using a bootstrap argument, which utilizes the continuity of \eqref{eq:bootstrap}
and smallness of the linear part of \eqref{eq:bootstrap_lin} one
can show \cite{Fishman2009a} that until some time $t_{*}$ , the
dynamics of \eqref{eq:bootstrap} is governed by the linear part and
the remainder is bounded by, \begin{eqnarray}
\left|Q_{n}\left(t\right)\right| & \leq & A\cdot t\cdot e^{-\gamma\left\vert n\right\vert },\label{eq:t0_definition}\end{eqnarray}
where $\gamma$ is the inverse localization length. Therefore to estimate
the remainder we can integrate \eqref{eq:bootstrap_lin} instead of
\eqref{eq:bootstrap} at least up to $t_{*}$. It is useful to integrate
up to some large time, $t\ll t_{*},$ and then to extrapolate using
the linear bound \eqref{eq:t0_definition} up to $t_{*}$. In the
next section it will be proposed how to determine $t_{*}$ in practice.

\section{\label{sec:Results}Results}

In this section it will be demonstrated, how the numerical scheme
for calculations in the framework of the perturbation theory, is implemented
in practice. Some results will be compared with an exact numerical
solution of the original equation \eqref{eq:NLSE}.

For this purpose we have calculated numerically all the $c_{n}^{\left(l\right)}$
and $E_{n}^{'}$ for $l\leq4$ for a certain realization of the random
potential. To compare perturbation theory results to the exact results,
we compute their Fourier transform for different orders of expansion.
On Fig. \ref{fig:FFT} we see the Fourier transform of $\bar{c}_{0}$,
$\bar{c}_{1}$ and $\bar{c}_{9}$ (see \eqref{eq:c_tilde_def}) compared
to the Fourier transform of an exact (numerical) solution, $c_{n}$,
calculated using a split-step method. We notice a reasonable agreement
of the perturbation theory with an exact solution for $\bar{c}_{0}$,
$\bar{c}_{1}$ and a disagreement for $\bar{c}_{9}$ . %
\begin{figure}[tbh]
\subfloat{\includegraphics[width=0.7\textwidth]{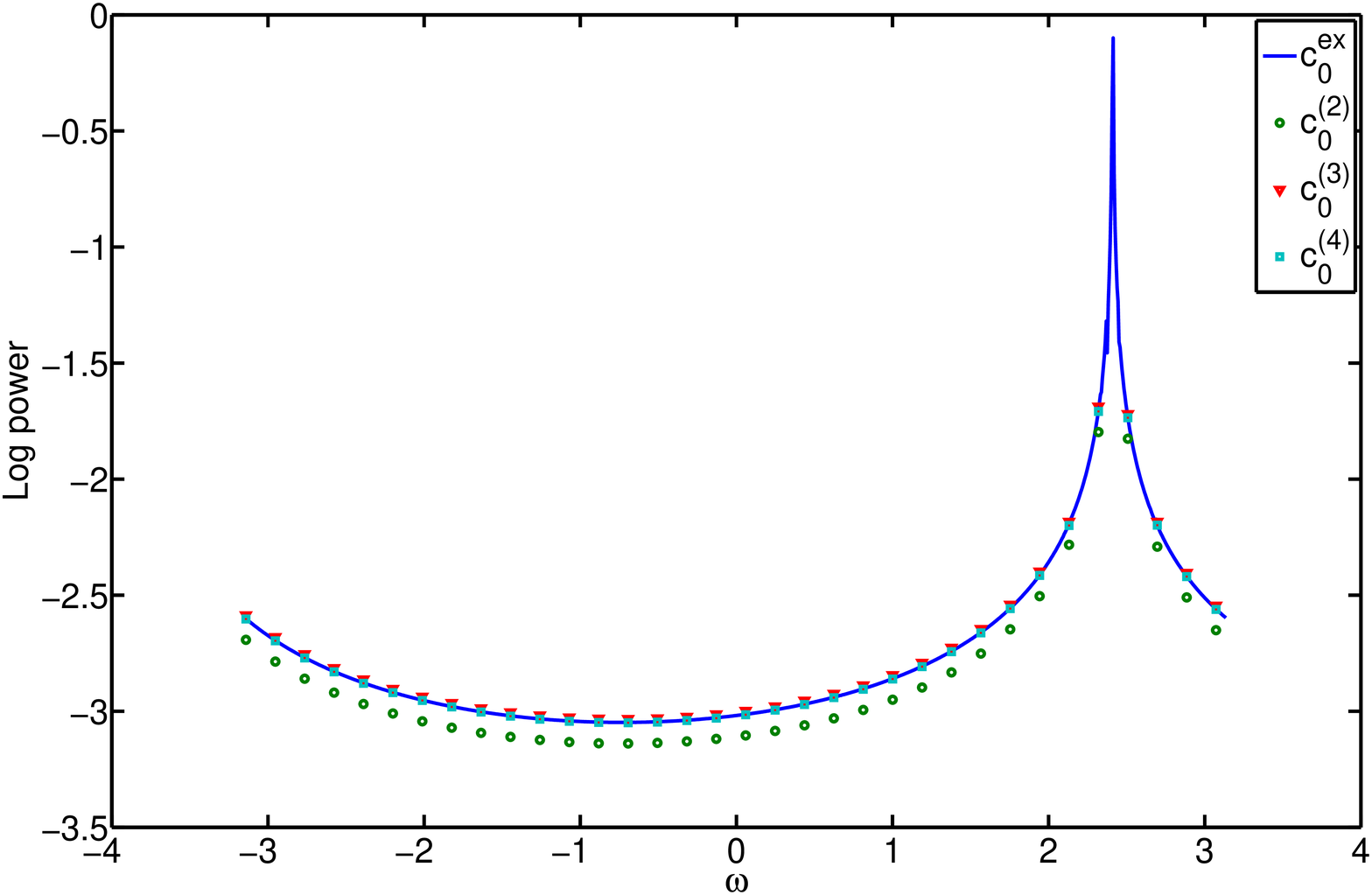}}\hfill{}\subfloat{\includegraphics[width=0.7\textwidth]{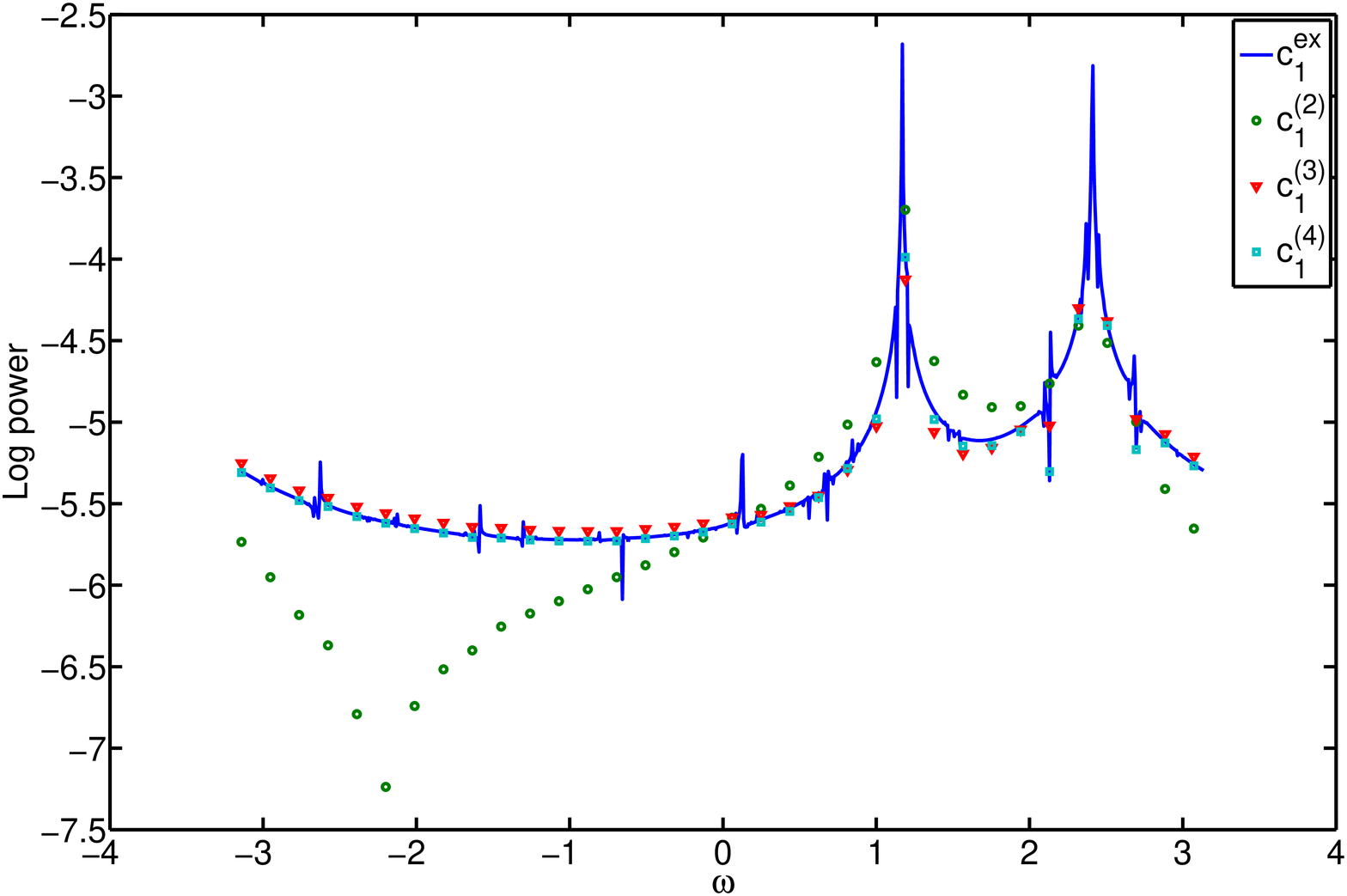}

}\hfill{}\subfloat{\includegraphics[width=0.7\linewidth]{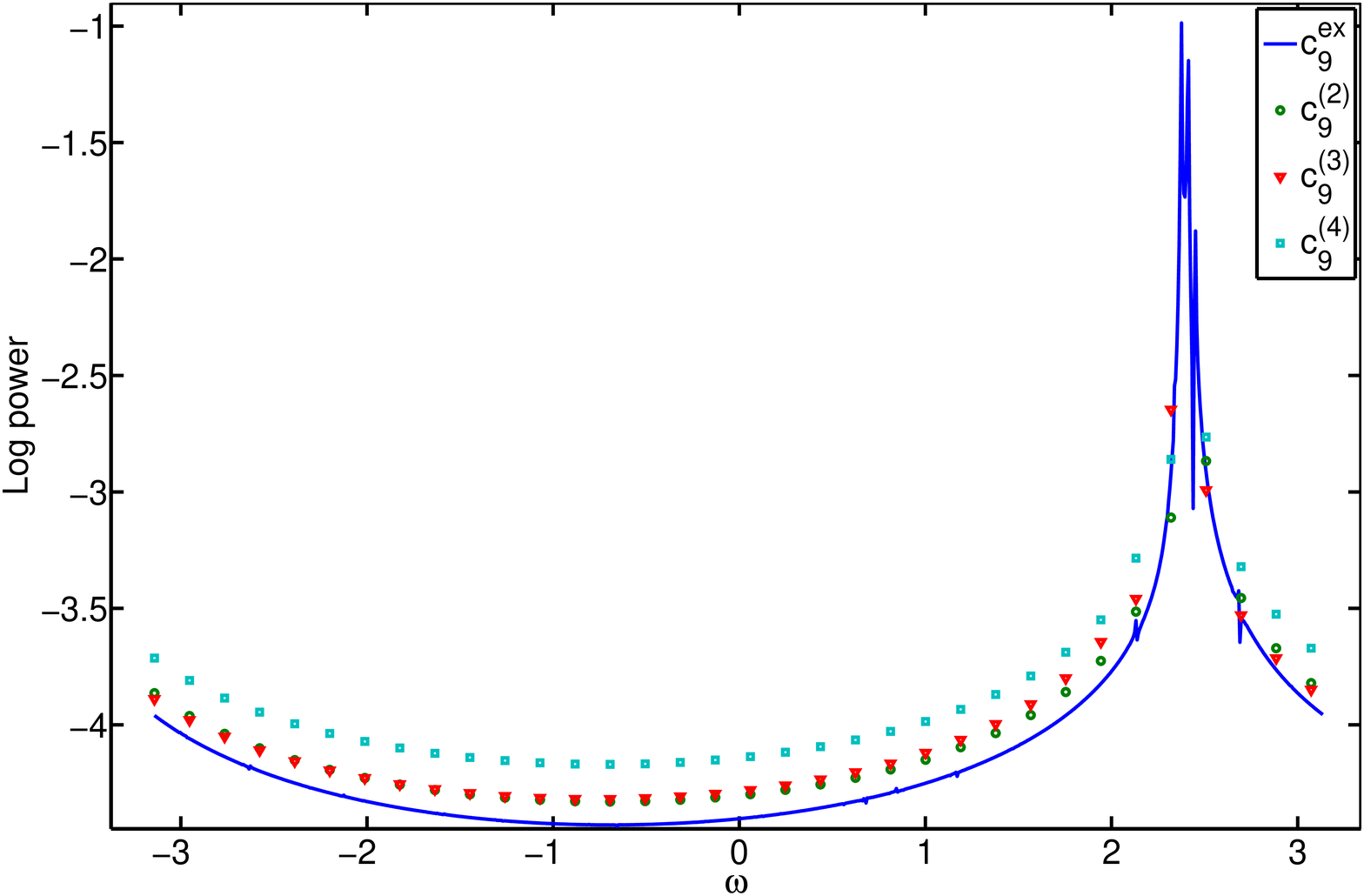}

}

\caption{\label{fig:FFT}Fourier transform of $\bar{c}_{0}$, $\bar{c}_{1}$
and $\bar{c}_{9}$ for different orders of the perturbation theory
compared to the Fourier transform of an exact solution, $c_{n}$,
(solid line). Second order is given by green crosses, third order
by red triangles and fourth order by blue squares. The parameters
are: $\beta=0.0774$, $t=1000$, $W=4$, $J=1$.}

\end{figure}
By plotting $Q_{n}^{lin}\left(t\right)$ for all $n$ with the same
scale (on the same axis) in Fig. \ref{fig:all_Qs}, we see that there
are modes (on Fig. \ref{fig:all_Qs} there are two of them) which
contribute to most of the discrepancy in the perturbation theory calculation,
since if $Q_{n}^{lin}\left(t\right)$ is large the bound on $Q_{n}\left(t\right)$
is also large.%
\begin{figure}[H]
\includegraphics[width=0.8\textwidth]{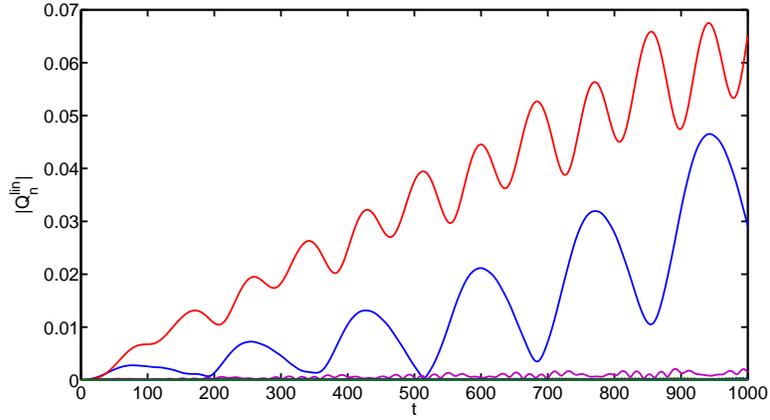}\caption{\label{fig:all_Qs}$Q_{n}^{lin}$ as a function of time in the 4th
order in $\beta$ for all $n$'s of the lattice (total 128 lines).
The two lines that are far above the rest (which are barely visible)
correspond to the resonant modes, $n=4,9$. The parameters are: $\beta=0.0774$,
$W=4$, $J=1$.}

\end{figure}
 We will call those modes resonant modes. In Fig. \ref{fig:Q_Qbar}
we compare a norm of $Q_{n}$, calculated with the resonant modes,
\begin{equation}
\left\Vert Q\right\Vert _{2}=\left(\sum_{m}\left|Q_{m}\right|^{2}\right)^{1/2},\end{equation}
and without them $\left\Vert Q'\right\Vert _{2}$. %
\begin{figure}[H]
\includegraphics[width=0.8\textwidth]{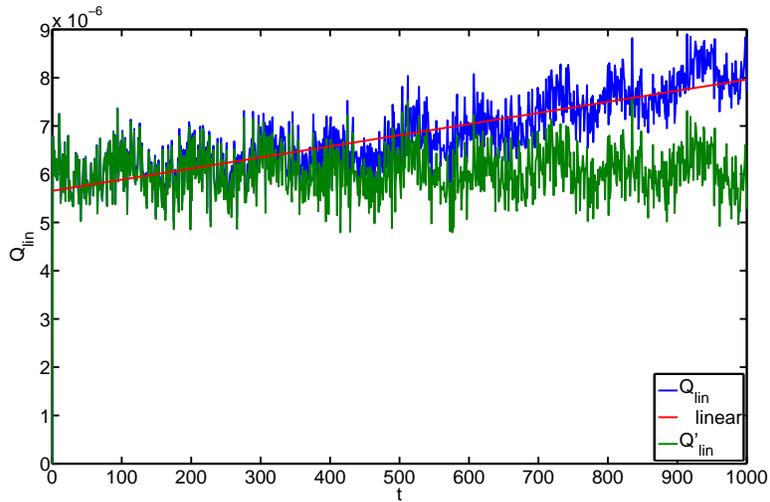}\caption{\label{fig:Q_Qbar}$\left\Vert Q\right\Vert _{2}$ (dashed blue) and
$\left\Vert Q'\right\Vert _{2}$ (solid green) as a function of time.
The straight line is a linear fit to $\left\Vert Q\right\Vert _{2}$.
The parameters are: 4th order, $\beta=0.01$, $W=4$, $J=1$. The
linear fit is: $\left\Vert Q\right\Vert _{2}\left(t\right)=2.3\times10^{-9}\cdot t+5.7\times10^{-6}$. }

\end{figure}
 It is found that indeed the $\left\Vert Q\right\Vert _{2}$ grows
linearly with time as expected from \eqref{eq:t0_definition}. Actually,
this is the way a bound on the resonant terms is expected to behave
\cite{Fishman2009a}. It is evident that by excluding the resonant
modes the discrepancy between the perturbation theory and exact results
is much lower. The reason why the perturbation theory fails for the
resonant modes is that they correspond to a quasi-degeneracy, namely,
when $E_{n}^{\prime}\thickapprox E_{0}^{\prime}$, and the overlap
$\left|V_{n}^{000}\right|$ is not sufficiently small. A natural way
to quantify the resonance condition is to use (compare to Eq.(5) of
\cite{Flach2009}),\begin{equation}
R_{n}^{-1}\equiv\left|\frac{V_{n}^{000}}{E_{n}^{\prime}-E_{0}^{\prime}}\right|.\label{eq:R}\end{equation}
The resonant modes produce substantially higher values of $R_{n}^{-1}$
compared to any other modes as demonstrated in Fig. \ref{fig:R}.
\begin{figure}[H]
\includegraphics[width=0.8\textwidth]{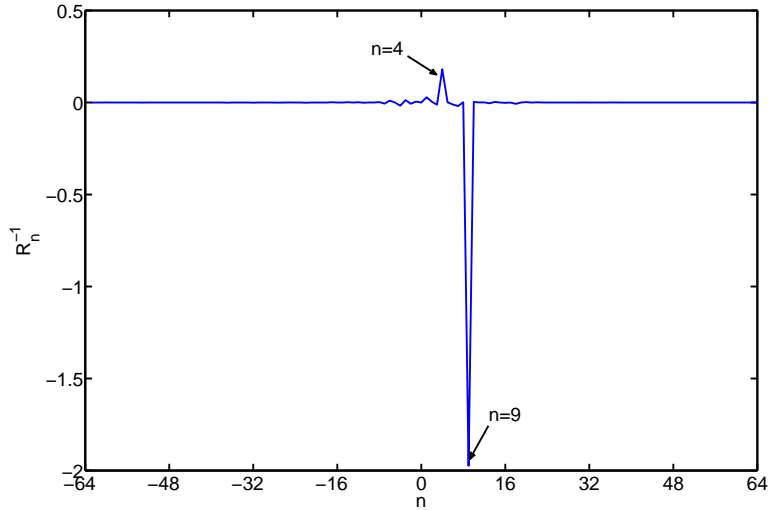}\caption{\label{fig:R}$R_{n}^{-1}$ as a function of $n$}

\end{figure}
 One way to deal with this problem is by using a degenerate perturbation
theory. However it is an open question how to implement it for a nonlinear
problem, that should be left for further studies.

At the end of the last section we have claimed that for some time,
$t_{*},$ the linear part of \eqref{eq:bootstrap} dominates over
the nonlinear part given that the linear part is sufficiently small.
In Fig. \ref{fig:Qex_Qlin} we present a comparison between the solution
of the linear equation \eqref{eq:bootstrap_lin} and $Q_{ex}$, the
solution of \eqref{eq:bootstrap}. It is clear that until $\left\Vert Q_{lin}\right\Vert _{2}\sim0.1$
both the solutions are very close.%
\begin{figure}[H]
\includegraphics[width=0.8\textwidth]{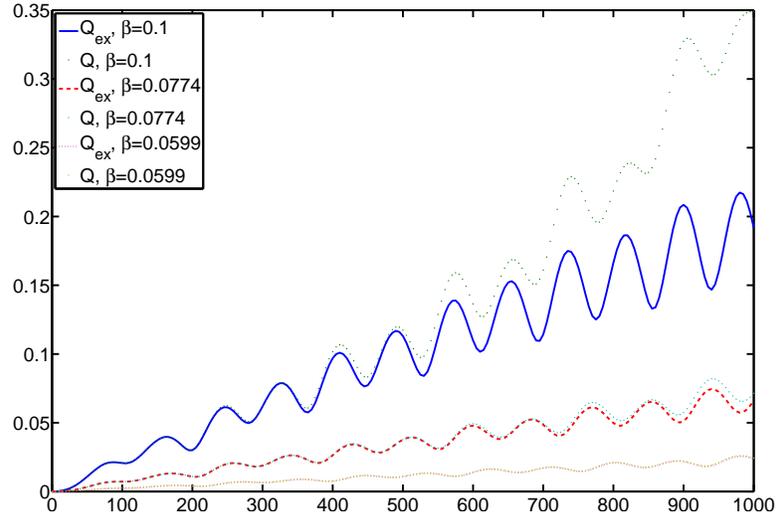}\caption{\label{fig:Qex_Qlin}$\left\Vert Q_{ex}\right\Vert _{2}$ (lines)
and $\left\Vert Q_{lin}\right\Vert _{2}$ (dotes) as a function of
time, for various values of $\beta$ (see legend). The perturbation
expansion is up to fourth order in $\beta$.}

\end{figure}
 We use this value to define $t_{*}$,\begin{equation}
\left\Vert Q_{lin}\left(t_{*}\right)\right\Vert _{2}=0.1.\label{eq:t_star_def}\end{equation}
For small nonlinearity strength, $\beta$, $t_{*}$ is very large
and therefore the integration of \eqref{eq:bootstrap_lin} to $t_{*}$
is very time consuming. We therefore use the bound \eqref{eq:t0_definition}
to extrapolate linearly from the time interval where \eqref{eq:bootstrap_lin}
is solved to $t_{*}$. Practically, we have calculated the linear
behavior of $Q_{lin}$ like it was done in Fig. \ref{fig:Q_Qbar}
than we found $t_{*}$ from \eqref{eq:t_star_def}. In Fig. \ref{fig:t_star}
we plot $\log_{10}t_{*}$ as a function of $\beta^{-1}$ for different
orders, while in Fig. \ref{fig:t_star_prime} we compare the result
found in 4-th order with $\left|Q_{n}^{lin}\right|$ of \eqref{eq:t_star_def}
is replaced by $\left|Q_{n}^{lin'}\right|$ (where the resonant terms
are removed) %
\begin{figure}[H]
\includegraphics[width=0.8\textwidth]{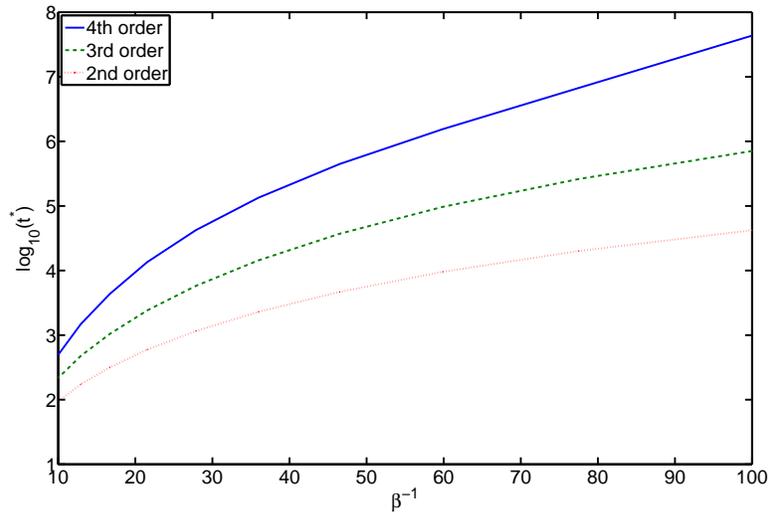}\caption{\label{fig:t_star}$\log_{10}t_{*}$ as a function of $\beta^{-1}$
for different orders. 4th order (solid blue), 3rd order (dashed green)
and 2nd order (dotted red). The parameters are: $W=4$, $J=1$.}

\end{figure}
\begin{figure}[H]
\includegraphics[width=0.8\textwidth]{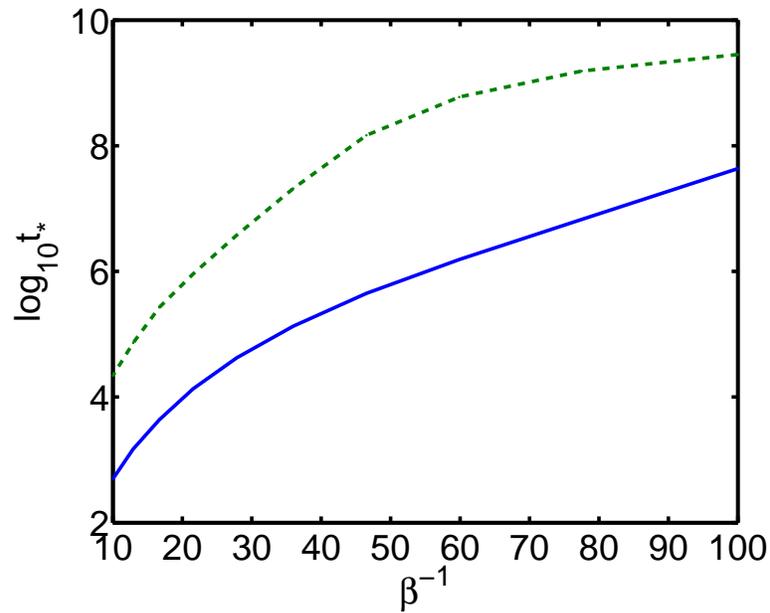}\caption{\label{fig:t_star_prime}$\log_{10}t_{*}$ as a function of $\beta^{-1}$
for the 4-th order.. Solid blue line is when $t_{*}$ is calculated
based on $\left|Q_{n}^{lin}\right|$ and dashed green line is when
$\left|Q_{n}^{lin'}\right|$ is used instead. The parameters are:
$W=4$, $J=1$.}

\end{figure}
 A systematic improvement with the order of the perturbation theory
is found. We notice that even with moderate nonlinearity strengths,
namely, $\beta<0.08$, one can achieve a good approximation of the
solution up to very large times. In this way the perturbation theory
combined with the solution of the linear equation \eqref{eq:bootstrap_lin}
and the criterion \eqref{eq:t_star_def} may be used to obtain the
solution of the original equation \eqref{eq:NLSE} up to $t<t_{*}$.
For small $\beta$ the time $t_{*}$ is very long as is clear from
Fig. \ref{fig:t_star}. When one considers the smallness of $\beta$
one should consider actually the smallness of $\beta e^{2}$ due to
the exponential proliferation of the number of terms (see Eq. 4.6
of \cite{Fishman2009a}).

Assuming that the $Q_{n}'$ give a good approximaion of the $Q_{n}$
for $t<t_{*}$, we can use Fig. \ref{fig:t_star_prime} to conclude
that, for example, for $\beta=0.1$ our solution is meaningful up
to time $t=10^{4.3}$ and therefore can be whithin reach of works
like \cite{Flach2009,Skokos2009}. The assumption obove corresponds
to the observation that correcting $c_{4}$ and $c_{9}$ by $Q_{4}$
and $Q_{9}$ is not changed much over time, and therefore all the
$Q_{n}^{'}$ stay small for a very long time. Note, that since the
largest localization length for $W=4$ is approximatelly 6, positions
4 and 9 are within one localization length from the initial data.
Therefore their removal does not affect the behavior for large $n$
that is relevant for the asymptotics.

Additional benifit of this implemamtnation is the computational speed.
For example, even for the 3rd order of the pertrurbation theory one
can compute the perturbative solution up to $t<10^{4}$ in 10 minutes
while using the same computer an exact convergent integration (split-step
method) takes 28 hours, which is a two orders of magnitude speed-up.

\section{\label{sec:Summary}Summary and Discussion}

In this paper we have demonstrated how a perturbation theory can be
numerically implemented for the solution of the NLSE with a random
potential \eqref{eq:NLSE}. In the perturbative method the computer
is used to implement the symbolic recursive calculation of the various
terms and also for their numerical evaluation. This method allows
to estimate the errors since the remainder term is bounded. It was
demonstrated (in the end of the previous section) how evaluate the
time of validity of the perturbation theory. This approach has also
a great advantage in the speed of the calculation. We believe that
the method can be generalized to other nonlinear differential equations,
for example, the Fermi-Ulam-Pasta problem.

In order to make the perturbation theory of a much greater value one
has to solve the problem of resonant terms (when $R_{n}^{-1}$ of
\eqref{eq:R} is large). In other words the degenerate perturbation
theory should be extended to nonlinear equations. Also the asymptotic
nature of the perturbation theory is not known, it is not clear when
it is convergent and when only asymptotic.

We enjoyed many extensive illuminating and extremely critical discussions
with Michael Aizenman. We also had informative discussions with S.
Aubry, V. Chulaevski, S. Flach, I. Goldshield, M. Goldstein, I. Guarneri,
M. Sieber, W.-M. Wang and S. Warzel. This work was partly supported
by the Israel Science Foundation (ISF), by the US-Israel Binational
Science Foundation (BSF), by the USA National Science Foundation (NSF
DMS-0903651), by the Minerva Center of Nonlinear Physics of Complex
Systems, by the Shlomo Kaplansky academic chair, by the Fund for promotion
of research at the Technion and by the E. and J. Bishop research fund.

\section*{Appendix}

In the calculations of the paper we used the expansion \eqref{eq:expansion_prime}
of the wavefunction and the expansion coefficients $c_{n}\left(t\right)$
are calculated from \eqref{eq:cn_expand} and \eqref{eq:a2}. The
$c_{n}^{\left(l\right)}\left(t\right)$ depend on $\beta$ since the
$E_{n}^{'}$ in \eqref{eq:a2} depend on it. Therefore \eqref{eq:cn_expand}
is not a Taylor series in $\beta$. In this Appendix it is shown that
this expansion is equivalent to the Taylor series,\[
c_{n}=\sum_{l=0}^{\infty}\bar{c}_{n}^{\left(l\right)}\beta^{l}.\]
Contrary to the $c_{n}^{\left(l\right)}$ the $\bar{c}_{n}^{\left(l\right)}$
are independent of $\beta$. The expansion of \eqref{eq:diff_eq}
in powers of $\beta$ produces the following equation for the $r$-th
order\begin{align}
i\partial_{t}\bar{c}_{n}^{\left(r\right)} & =-\sum_{l=1}^{r}\sum_{l_{1}=0}^{l}\frac{1}{l_{1}!}\left(\frac{\partial^{l_{1}}E_{n}^{\left(l-l_{1}\right)}}{\partial\beta^{l_{1}}}\right)_{\beta=0}\bar{c}_{n}^{\left(r-l\right)}+\label{eq:c}\\
 & +\sum_{\sum_{i}l_{i}=r-1}\sum_{\left\{ m_{i}\right\} }V_{n}^{m_{1}m_{2}m_{3}}\bar{c}_{m_{1}}^{*\left(l_{1}\right)}\bar{c}_{m_{2}}^{\left(l_{2}\right)}\bar{c}_{m_{3}}^{\left(l_{3}\right)}\left(\frac{1}{l_{4}!}\frac{\partial^{l_{4}}}{\partial\beta^{l_{4}}}e^{i\left(E'_{n}+E'_{m_{1}}-E'_{m_{2}}-E'_{m_{3}}\right)t}\right)_{\beta=0}.\nonumber \end{align}
We considered a different expansion \eqref{eq:a2}, where one does
not expand the exponent, but still compares the implicit powers of
$\beta$ from both sides of the equation. The equation for the $\left(r-k\right)$
order in this expansion is (that is just \eqref{eq:a2}) \begin{align}
i\partial_{t}c_{n}^{\left(r-k\right)} & =-\sum_{l=1}^{r-k}E_{n}^{\left(l\right)}c_{n}^{\left(r-k-l\right)}\label{eq:c_bar}\\
 & +\sum_{\left\{ m_{i}\right\} }\sum_{l_{1}+l_{2}+l_{3}=r-k-1}V_{n}^{m_{1}m_{2}m_{3}}c_{m_{1}}^{*\left(l_{1}\right)}c_{m_{2}}^{\left(l_{2}\right)}c_{m_{3}}^{\left(l_{3}\right)}e^{i\left(E'_{n}+E'_{m_{1}}-E'_{m_{2}}-E'_{m_{3}}\right)t}.\nonumber \end{align}
It can be shown that each term in this expansion is uniformly bounded
in time, with a proper selection of $E_{n}'$, that is the secular
terms are removed. We will show that this expansion is equivalent
to a Taylor series.
\begin{thm}
For any $t,$ \begin{equation}
c_{n}=\sum_{k=0}^{\infty}c_{n}^{\left(k\right)}\beta^{k},\end{equation}
where $c_{n}^{\left(k\right)}$ are defined using \eqref{eq:c_bar}.\end{thm}
\begin{proof}
Since $c_{n}=\sum_{k=0}^{\infty}\bar{c}_{n}^{\left(k\right)}\beta^{k}$
is the expansion of $c_{n}$ in powers of $\beta$, and due to the
uniqueness of this expansion, the theorem is true if and only if\begin{equation}
\bar{c}_{n}^{\left(r\right)}=\sum_{k=0}^{r}\frac{1}{k!}\left(\frac{\partial^{k}c_{n}^{\left(r-k\right)}}{\partial\beta^{k}}\right)_{\beta=0}=\sum_{k=0}^{r-1}\frac{1}{k!}\left(\frac{\partial^{k}c_{n}^{\left(r-k\right)}}{\partial\beta^{k}}\right)_{\beta=0},\end{equation}
where at the last equality we have used the fact that $\frac{\partial^{r}c_{n}^{\left(0\right)}}{\partial\beta^{r}}=0,$
for any $r\geq1.$ We proceed to prove this equality by induction.
Suppose that this equality is true for all the orders till order $r$.
We apply the linear operator $\sum_{k=0}^{r-1}\frac{1}{k!}\frac{\partial^{k}}{\partial\beta^{k}}$
to both sides of \eqref{eq:c_bar}. For the first part on the RHS
we get\begin{align}
\sum_{k=0}^{r-1}\frac{1}{k!}\sum_{l=1}^{r-k}\frac{\partial^{k}}{\partial\beta^{k}}\left(E_{n}^{\left(l\right)}c_{n}^{\left(r-k-l\right)}\right) & =\nonumber \\
 & =\sum_{k=0}^{r-1}\frac{1}{k!}\sum_{l=1}^{r-k}\sum_{l_{1}=0}^{k}\frac{k!}{l_{1}!\left(k-l_{1}\right)!}\frac{\partial^{l_{1}}E_{n}^{\left(l\right)}}{\partial\beta^{l_{1}}}\frac{\partial^{\left(k-l_{1}\right)}c_{n}^{\left(r-k-l\right)}}{\partial\beta^{\left(k-l_{1}\right)}}\nonumber \\
 & =\sum_{k=0}^{r-1}\sum_{l=1}^{r-k}\sum_{l_{1}=0}^{k}\left(\frac{1}{l_{1}!}\frac{\partial^{l_{1}}E_{n}^{\left(l\right)}}{\partial\beta^{l_{1}}}\right)\left(\frac{1}{\left(k-l_{1}\right)!}\frac{\partial^{\left(k-l_{1}\right)}c_{n}^{\left(r-k-l\right)}}{\partial\beta^{\left(k-l_{1}\right)}}\right)\label{eq:first_sum}\end{align}
where at the first equality we have used the Leibniz generalized product
rule. We now exchange variables such that\begin{eqnarray}
l_{1} & = & l_{1}\\
l_{2} & = & k-l_{1}\nonumber \\
z & = & l+l_{1}\nonumber \end{eqnarray}
or in a matrix notation\begin{equation}
\left(\begin{array}{c}
l_{1}\\
l_{2}\\
z\end{array}\right)=A_{1}\left(\begin{array}{c}
l_{1}\\
l\\
k\end{array}\right),\end{equation}
where\[
A_{1}=\left(\begin{array}{ccc}
1 & 0 & 0\\
-1 & 0 & 1\\
1 & 1 & 0\end{array}\right).\]
The region of summation is bounded by the planes\begin{eqnarray}
0 & \leq & k\leq r-1\\
1 & \leq & l\leq r-k\nonumber \\
0 & \leq & l_{1}\leq k\nonumber \end{eqnarray}
or in matrix notation\begin{equation}
B_{1}\left(\begin{array}{c}
l_{1}\\
l\\
k\end{array}\right)\leq\left(\begin{array}{c}
r-1\\
r\\
0\\
0\\
-1\\
0\end{array}\right),\end{equation}
where\[
B_{1}=\left(\begin{array}{ccc}
0 & 0 & 1\\
0 & 1 & 1\\
1 & 0 & -1\\
-1 & 0 & 0\\
0 & -1 & 0\\
0 & 0 & -1\end{array}\right).\]
Changing the variables by applying the transformation matrix produces\begin{equation}
B_{1}\cdot A_{1}^{-1}\left(\begin{array}{c}
l_{1}\\
l_{2}\\
z\end{array}\right)=\left(\begin{array}{ccc}
1 & 1 & 0\\
0 & 1 & 1\\
0 & -1 & 0\\
-1 & 0 & 0\\
1 & 0 & -1\\
-1 & -1 & 0\end{array}\right)\left(\begin{array}{c}
l_{1}\\
l_{2}\\
z\end{array}\right)\leq\left(\begin{array}{c}
r-1\\
r\\
0\\
0\\
-1\\
0\end{array}\right)\end{equation}
or the following inequalities\begin{eqnarray}
l_{1}+l_{2} & \leq & r-1\\
l_{2}+z & \leq & r\nonumber \\
0 & \leq & l_{2}\nonumber \\
0 & \leq & l_{1}\nonumber \\
l_{1} & \leq & z-1\nonumber \\
0 & \leq & l_{1}+l_{2}.\nonumber \end{eqnarray}
Removing redundant inequalities gives\begin{eqnarray}
0\leq l_{1} & \leq & z-1\\
1 & \leq & z\leq r\nonumber \\
0 & \leq & l_{2}\leq r-z.\nonumber \end{eqnarray}
Therefore in the new variables the sum \eqref{eq:first_sum} takes
the form\begin{equation}
\sum_{z=1}^{r}\sum_{l_{1}=0}^{z-1}\sum_{l_{2}=0}^{r-z}\left(\frac{1}{l_{1}!}\frac{\partial^{l_{1}}E_{n}^{\left(z-l_{1}\right)}}{\partial\beta^{l_{1}}}\right)\left(\frac{1}{l_{2}!}\frac{\partial^{l_{2}}c_{n}^{\left(r-z-l_{2}\right)}}{\partial\beta^{l_{2}}}\right),\end{equation}
 Since $\frac{\partial^{z}E_{n}^{\left(0\right)}}{\partial\beta^{z}}=0$
for any $z\geq1,$we can write\begin{equation}
\sum_{z=1}^{r}\sum_{l_{1}=0}^{z}\sum_{l_{2}=0}^{r-z}\left(\frac{1}{l_{1}!}\frac{\partial^{l_{1}}E_{n}^{\left(z-l_{1}\right)}}{\partial\beta^{l_{1}}}\right)\left(\frac{1}{l_{2}!}\frac{\partial^{l_{2}}c_{n}^{\left(r-z-l_{2}\right)}}{\partial\beta^{l_{2}}}\right).\end{equation}
 Taking $\beta=0$ and using the assumption of the induction for orders
lower than $r$ we have\begin{equation}
\sum_{z=1}^{r}\sum_{l_{1}=0}^{z}\left(\frac{1}{l_{1}!}\frac{\partial^{l_{1}}E_{n}^{\left(z-l_{1}\right)}}{\partial\beta^{l_{1}}}\right)\bar{c}_{n}^{\left(r-z\right)},\end{equation}
which is the first expression for $i\partial_{t}\bar{c}_{n}^{\left(r\right)}$.
The proof for the second term is similar, operating with $\sum_{k=0}^{r-1}\frac{1}{k!}\frac{\partial^{k}}{\partial\beta^{k}}$
on the second term in \eqref{eq:c_bar} gives \begin{equation}
\sum_{\left\{ m_{i}\right\} }V_{n}^{m_{1}m_{2}m_{3}}\sum_{k=0}^{r-1}\sum_{l_{1}=0}^{r-k-1}\sum_{l_{2}=0}^{r-k-1-l_{1}}\frac{1}{k!}\frac{\partial^{k}}{\partial\beta^{k}}\left(c_{m_{1}}^{*\left(l_{1}\right)}c_{m_{2}}^{\left(l_{2}\right)}c_{m_{3}}^{\left(r-k-1-l_{1}-l_{2}\right)}e^{i\left(E'_{n}+E'_{m_{1}}-E'_{m_{2}}-E'_{m_{3}}\right)t}\right).\end{equation}
 Using the Leibniz generalized product rule, which states\begin{equation}
\frac{\partial^{k}}{\partial\beta^{k}}\left(x_{1}x_{2}x_{3}x_{4}\right)=\sum_{s_{1}+s_{2}+s_{3}+s_{4}=k}\frac{k!}{s_{1}!s_{2}!s_{3}!s_{4}!}\frac{\partial^{s_{1}}x_{1}}{\partial\beta^{s_{1}}}\frac{\partial^{s_{2}}x_{2}}{\partial\beta^{s_{2}}}\frac{\partial^{s_{3}}x_{3}}{\partial\beta^{s_{3}}}\frac{\partial^{s_{4}}x_{4}}{\partial\beta^{s_{4}}}\end{equation}
we have the sum\begin{align}
 & \sum_{\left\{ m_{i}\right\} }V_{n}^{m_{1}m_{2}m_{3}}\sum_{k=0}^{r-1}\sum_{l_{1}=0}^{r-k-1}\sum_{l_{2}=0}^{r-k-1-l_{1}}\sum_{s_{1}=0}^{k}\sum_{s_{2}=0}^{k-s_{1}}\sum_{s_{3}=0}^{k-s_{1}-s_{2}}\left(\frac{1}{s_{1}!}\frac{\partial^{s_{1}}c_{m_{1}}^{*\left(l_{1}\right)}}{\partial\beta^{s_{1}}}\right)\left(\frac{1}{s_{2}!}\frac{\partial^{s_{2}}c_{m_{2}}^{\left(l_{2}\right)}}{\partial\beta^{s_{2}}}\right)\times\nonumber \\
 & \times\left(\frac{1}{s_{3}!}\frac{\partial^{s_{3}}c_{m_{3}}^{\left(r-k-1-l_{1}-l_{2}\right)}}{\partial\beta^{s_{3}}}\right)\left(\frac{1}{\left(k-s_{1}-s_{2}-s_{3}\right)!}\frac{\partial^{\left(k-s_{1}-s_{2}-s_{3}\right)}}{\partial\beta^{\left(k-s_{1}-s_{2}-s_{3}\right)}}e^{i\left(E'_{n}+E'_{m_{1}}-E'_{m_{2}}-E'_{m_{3}}\right)t}\right)\end{align}
Following the first part of this proof we exchange the variable to
\begin{eqnarray}
z_{1} & = & l_{1}+s_{1}\\
z_{2} & = & l_{2}+s_{2}\nonumber \\
z_{3} & = & k-s_{1}-s_{2}-s_{3}\nonumber \\
s_{i} & = & s_{i}\nonumber \end{eqnarray}
or using a transformation matrix\begin{equation}
\left(\begin{array}{c}
z_{1}\\
z_{2}\\
z_{3}\\
s_{1}\\
s_{2}\\
s_{3}\end{array}\right)=A_{2}\left(\begin{array}{c}
k\\
l_{1}\\
l_{2}\\
s_{1}\\
s_{2}\\
s_{3}\end{array}\right),\end{equation}
where\[
A_{2}=\left(\begin{array}{cccccc}
0 & 1 & 0 & 1 & 0 & 0\\
0 & 0 & 1 & 0 & 1 & 0\\
1 & 0 & 0 & -1 & -1 & -1\\
0 & 0 & 0 & 1 & 0 & 0\\
0 & 0 & 0 & 0 & 1 & 0\\
0 & 0 & 0 & 0 & 0 & 1\end{array}\right).\]
The region of summation is bounded by the following hyperplanes\begin{eqnarray}
0 & \leq & k\leq r-1\\
0 & \leq & l_{1}\leq r-k-1\nonumber \\
0 & \leq & l_{2}\leq r-k-1-l_{1}\nonumber \\
0 & \leq & s_{1}\leq k\nonumber \\
0 & \leq & s_{2}\leq k-s_{1}\nonumber \\
0 & \leq & s_{3}\leq k-s_{1}-s_{2}\nonumber \end{eqnarray}
which could be represented in a matrix notation as\begin{equation}
B_{2}\left(\begin{array}{c}
k\\
l_{1}\\
l_{2}\\
s_{1}\\
s_{2}\\
s_{3}\end{array}\right)\leq\left(\begin{array}{c}
r-1\\
r-1\\
r-1\\
0\\
0\\
0\\
0\\
0\\
0\\
0\\
0\\
0\end{array}\right),\end{equation}
where\[
B_{2}=\left(\begin{array}{cccccc}
1 & 0 & 0 & 0 & 0 & 0\\
1 & 1 & 0 & 0 & 0 & 0\\
1 & 1 & 1 & 0 & 0 & 0\\
-1 & 0 & 0 & 1 & 0 & 0\\
-1 & 0 & 0 & 1 & 1 & 0\\
-1 & 0 & 0 & 1 & 1 & 1\\
-1 & 0 & 0 & 0 & 0 & 0\\
0 & -1 & 0 & 0 & 0 & 0\\
0 & 0 & -1 & 0 & 0 & 0\\
0 & 0 & 0 & -1 & 0 & 0\\
0 & 0 & 0 & 0 & -1 & 0\\
0 & 0 & 0 & 0 & 0 & -1\end{array}\right).\]
Using the transformation matrix to change the variables results in
the following \begin{equation}
\left(\begin{array}{cccccc}
0 & 0 & 1 & 1 & 1 & 1\\
1 & 0 & 1 & 0 & 1 & 1\\
1 & 1 & 1 & 0 & 0 & 1\\
0 & 0 & -1 & 0 & -1 & -1\\
0 & 0 & -1 & 0 & 0 & -1\\
0 & 0 & -1 & 0 & 0 & 0\\
0 & 0 & -1 & -1 & -1 & -1\\
-1 & 0 & 0 & 1 & 0 & 0\\
0 & -1 & 0 & 0 & 1 & 0\\
0 & 0 & 0 & -1 & 0 & 0\\
0 & 0 & 0 & 0 & -1 & 0\\
0 & 0 & 0 & 0 & 0 & -1\end{array}\right)\left(\begin{array}{c}
z_{1}\\
z_{2}\\
z_{3}\\
s_{1}\\
s_{2}\\
s_{3}\end{array}\right)\leq\left(\begin{array}{c}
r-1\\
r-1\\
r-1\\
0\\
0\\
0\\
0\\
0\\
0\\
0\\
0\\
0\end{array}\right)\end{equation}
or in the following inequalities\begin{eqnarray}
z_{3}+s_{1}+s_{2}+s_{3} & \leq & r-1\\
z_{1}+z_{3}+s_{2}+s_{3} & \leq & r-1\nonumber \\
z_{1}+z_{2}+z_{3}+s_{3} & \leq & r-1\nonumber \\
0 & \leq & z_{3}+s_{2}+s_{3}\nonumber \\
0 & \leq & z_{3}+s_{3}\nonumber \\
0 & \leq & z_{3}\nonumber \\
0 & \leq & z_{3}+s_{1}+s_{2}+s_{3}\nonumber \\
s_{1} & \leq & z_{1}\nonumber \\
s_{2} & \leq & z_{2}\nonumber \\
0 & \leq & s_{i}\nonumber \end{eqnarray}
 Removing redundant inequalities gives\begin{eqnarray}
0 & \leq & s_{1}\leq z_{1}\\
0 & \leq & s_{2}\leq z_{2}\nonumber \\
0 & \leq & s_{3}\leq\left(r-1\right)-z_{1}-z_{2}-z_{3}\nonumber \\
0 & \leq & z_{1}\leq\left(r-1\right)-z_{3}\nonumber \\
0 & \leq & z_{2}\leq\left(r-1\right)-z_{1}-z_{3}\nonumber \\
0 & \leq & z_{3}\leq r-1\nonumber \end{eqnarray}
which is equivalent to the sum\begin{align}
 & \sum_{\left\{ m_{i}\right\} }V_{n}^{m_{1}m_{2}m_{3}}\sum_{z_{3}=0}^{r-1}\sum_{z_{1}=0}^{r-1-z_{3}}\sum_{z_{2}=0}^{r-1-z_{1}-z_{3}}\sum_{s_{1}=0}^{z_{1}}\sum_{s_{2}=0}^{z_{2}}\sum_{s_{3}=0}^{r-1-z_{1}-z_{2}-z_{3}}\left(\frac{1}{s_{1}!}\frac{\partial^{s_{1}}c_{m_{1}}^{*\left(z_{1}-s_{1}\right)}}{\partial\beta^{s_{1}}}\right)\times\nonumber \\
 & \times\left(\frac{1}{s_{2}!}\frac{\partial^{s_{2}}c_{m_{2}}^{\left(z_{2}-s_{2}\right)}}{\partial\beta^{s_{2}}}\right)\left(\frac{1}{s_{3}!}\frac{\partial^{s_{3}}c_{m_{3}}^{\left(r-1-z_{1}-z_{2}-z_{3}-s_{3}\right)}}{\partial\beta^{s_{3}}}\right)\left(\frac{1}{z_{3}!}\frac{\partial^{z_{3}}}{\partial\beta^{z_{3}}}e^{i\left(E'_{n}+E'_{m_{1}}-E'_{m_{2}}-E'_{m_{3}}\right)t}\right)\end{align}
Putting $\beta=0$ and utilizing the assumption of the induction we
get\begin{equation}
\sum_{l_{1}+l_{2}+l_{3}+l_{4}=r-1}\sum_{\left\{ m_{i}\right\} }V_{n}^{m_{1}m_{2}m_{3}}\bar{c}_{m_{1}}^{\left(l_{1}\right)*}\bar{c}_{m_{2}}^{\left(l_{2}\right)}\bar{c}_{m_{3}}^{\left(l_{3}\right)}\left(\frac{1}{l_{4}!}\frac{\partial^{l_{4}}}{\partial\beta^{l_{4}}}e^{i\left(E'_{n}+E'_{m_{1}}-E'_{m_{2}}-E'_{m_{3}}\right)t}\right)_{\beta=0}\end{equation}
 which is exactly the second term in the equation for $i\partial_{t}\bar{c}_{n}^{\left(r\right)}$.
Since the zero order trivially satisfies this theorem, this competes
the proof by induction.
\end{proof}
One should be able to extend the results of this Appendix to other
nonlinear equations, for example, where the power of the nonlinearity
is different (in the last term of \eqref{eq:NLSE} 2 is replaced by
another integer).

\end{document}